\documentclass[conference, letterpaper]{IEEEtran}
\usepackage{epsfig}
\usepackage{times}
\usepackage{float}
\usepackage{afterpage}
\usepackage{amsmath}
\usepackage{amstext}
\usepackage{amssymb,bm}
\usepackage{latexsym}
\usepackage{color}
\usepackage{graphicx}
\usepackage{amsmath}
\usepackage{amsthm}
\usepackage{graphicx}
\usepackage[center]{caption}
\usepackage{subfigure}
\usepackage{booktabs}
\usepackage{multicol}
\usepackage{lipsum}
\usepackage{dblfloatfix}
\usepackage{mathrsfs}
\usepackage{cite}
\usepackage{tikz}
\usepackage{pgfplots}
\pgfplotsset{compat=newest}

\allowdisplaybreaks

\newcommand{\rmd}{\mathrm{d}}
\newcommand{\bbE}{\mathbb{E}}\newcommand{\rme}{\mathrm{e}}

\newcommand{\bbR}{\mathbb{R}}

\newcommand{\sfD}{\mathsf{D}}

\newcommand{\bfI}{\mathbf{I}}\newcommand{\sfI}{\mathsf{I}}

\newcommand{\bfN}{\mathbf{N}}

\newcommand{\sfR}{\mathsf{R}}

\newcommand{\bfW}{\mathbf{W}}
\newcommand{\bfX}{\mathbf{X}}\newcommand{\bfx}{\mathbf{x}}
\newcommand{\bfY}{\mathbf{Y}}\newcommand{\bfy}{\mathbf{y}}
\newcommand{\bfZ}{\mathbf{Z}}

\newcommand{\cB}{\mathcal{B}}
\newcommand{\cC}{\mathcal{C}}

\newcommand{\sfh}{\mathsf{h}}

\newcommand{\supp}{{\mathsf{supp}}}

\theoremstyle{mystyle}
\newtheorem{theorem}{Theorem}
\theoremstyle{mystyle}
\newtheorem{lemma}{Lemma}
\theoremstyle{mystyle}

\theoremstyle{mystyle}

\theoremstyle{mystyle}
\newtheorem{definition}{Definition}
\theoremstyle{remark}
\newtheorem{rem}{Remark}
\theoremstyle{mystyle}

\theoremstyle{mystyle}

\theoremstyle{mystyle}

\theoremstyle{discussion}

\theoremstyle{mystyle}

\theoremstyle{mystyle}

 \usepackage{enumitem}

\begin{document}

\title{On  the Capacity Achieving Input of Amplitude Constrained Vector Gaussian Wiretap~Channel }

\author{
\IEEEauthorblockN{Antonino Favano$^{* \dagger}$, Luca Barletta$^{*}$,  Alex Dytso$^{**}$}
$^{*}$ Politecnico di Milano, Milano, 20133, Italy. Email: $\{$antonino.favano, luca.barletta$\}$@polimi.it \\
$^{\dagger}$ Consiglio Nazionale delle Ricerche, Milano, 20133, Italy. \\
$^{**}$ New Jersey Institute of Technology, Newark,  NJ 07102, USA.
Email: alex.dytso@njit.edu}
\maketitle
\begin{abstract}
This paper studies secrecy-capacity of an $n$-dimensional Gaussian wiretap channel under the peak-power constraint. This work determines the largest peak-power constraint $\bar{\sfR}_n$ such that an input distribution uniformly distributed on a single sphere is optimal; this regime is termed the small-amplitude regime.  The asymptotic of $\bar{\sfR}_n$ as $n$ goes to infinity is completely characterized as a function of noise variance at both receivers.  Moreover, the secrecy-capacity is also characterized in a form amenable for computation. Furthermore, several numerical examples are provided, such as the example of the secrecy-capacity achieving distribution outside of the small amplitude regime.  
\end{abstract}

\section{Introduction} 
Consider the vector Gaussian wiretap channel with outputs 
\begin{align}
\bfY_1&= \bfX+\bfN_1,\\
\bfY_2&=\bfX+\bfN_2,
\end{align}
where  $\bfX \in \bbR^n$ and  where 
$\bfN_1 \sim \mathcal{N}(\mathbf{0}_n,\sigma_1^2 \bfI_n)$ and  $\bfN_2 \sim \mathcal{N}(\mathbf{0}_n,\sigma_2^2 \bfI_n)$, and with $(\bfX,\bfN_1,\bfN_2)$ mutually independent.  The output $\bfY_1$ is observed by the legitimate receiver whereas the output $\bfY_2$ is observed by the malicious receiver.  In this work, we are interested in the scenario where the input $\bfX$  is limited by a peak-power constraint or amplitude constraint and assume that $\bfX \in \cB_0(\sfR)$ where $  \cB_0(\sfR)$ is an $n$-ball centered at ${\bf 0}$ of radius $\sfR$.  
For this setting, the secrecy-capacity is given by 
\begin{align}
C_s(\sigma_1, \sigma_2, \sfR) &= \max_{\bfX \in  \cB_0(\sfR) }  I(\bfX; \bfY_1) - I(\bfX; \bfY_2) \\
&= \max_{\bfX \in  \cB_0(\sfR)}   I(\bfX; \bfY_1 | \bfY_2), \label{eq:Secracy_CAP}
\end{align}
where the last expression holds due to the degraded nature of the channel. 
It can be shown that for $\sigma_1^2  \ge \sigma^2_2$ the secrecy-capacity is equal to zero. Therefore, in the remaining, we assume that $\sigma_1^2 < \sigma^2_2$.

We are interested in studying the input distribution $P_{\bfX^\star}$ that maximizes \eqref{eq:Secracy_CAP} in the small (but not vanishing) amplitude regime.  Since closed-form expression for secrecy-capacity are rare, we are also interested in deriving an exact expression for the secrecy-capacity in this regime. We also argue in Section~\ref{sec:Connection_Other_Problem} the solution to the secrecy-capacity can shed light on other problems unrelated to security.

\subsection{Notation} 
The modified Bessel function of the first kind of order $v \ge 0 $ will be denoted by $\sfI_v(x), x\in \bbR$. The following ratio of the Bessel functions will be commonly used in this work:
\begin{equation}
\sfh_v(x) =\frac{\sfI_v(x)}{\sfI_{v-1}(x)}, x\in \bbR, v\ge 0. 
\end{equation}

We denote the distribution  of a random variable $\bfX$ by $P_{\bfX}$. The support set of $P_\bfX$ is denoted and defined as
\begin{align}
\supp(P_{\bfX})&=\{\bfx:  \text{ for every open set $ \mathcal{D} \ni \bfx $ } \notag\\ 
&\quad \qquad \text{
 we have that $P_{\bfX}( \mathcal{D})>0$} \}. 
\end{align} 
The minimum mean squared error is denoted by 
\begin{align}
{\rm mmse}(\bfX| \bfX+\bfN)=  \bbE \left[ \| \bfX-\bbE[\bfX| \bfX+\bfN] \|^2 \right].
  \end{align}

  \subsection{Literature Review} 
The wiretap channel was introduced by Wyner in \cite{wyner1975wire}, who also established the secrecy-capacity  of the degraded wiretap channel. The wiretap channel plays a central role in network information theory; the interested reader is referred to \cite{bloch2011physical,Oggier2015Wiretap,Liang2009Security,poor2017wireless} and reference therein for an in-detail treatment of the topic.

The secrecy-capacity of a scalar Gaussian wiretap channel with an average-power constraint was shown  in~\cite{GaussianWireTap} where the capacity-achieving input distribution was shown to be Gaussian.  The secrecy-capacity of the MIMO wiretap channel was characterized in \cite{khisti2010secure} and \cite{oggier2011secrecy} where the Gaussian input was shown to be optimal.  An elegant proof, using the I-MMSE relationship \cite{I-MMSE}, of optimality of Gaussian input, was given in \cite{bustin2009mmse}.

The secrecy-capacity of the Gaussian wiretap channel under the peak-power constraint has received far less attention.    The secrecy-capacity of the scalar Gaussian wiretap channel with an amplitude and power constraint was considered  in \cite{ozel2015gaussian} where the authors showed that the capacity-achieving input distribution $P_{X^\star}$ is discrete with finitely many support points. 
Recently, the result of  \cite{ozel2015gaussian}    was sharpened in \cite{barletta2021scalar} by providing an explicit upper bound on the number of support points of $P_{X^\star}$ of the following from: 
\begin{equation}
|  \supp(P_{X^\star}) | \le   \rho \frac{\sfR^2}{\sigma_1^2} + O( \log(\sfR) ), \label{eq:Upper_Bound_Explicit}
 \end{equation} 
where $\rho= (2\rme+1)^2 \left( \frac{\sigma_2+\sigma_1}{ \sigma_2-\sigma_1} \right)^2+ \left(\frac{\sigma_2+\sigma_1}{ \sigma_2-\sigma_1}+1 \right)^2$.  
The secrecy-capacity for the vector wiretap channel with a peak-power constraint was considered in \cite{DytsoITWwiretap2018} where it was shown that the optimal input distribution is concentrated on finitely many co-centric shells.

\section{Assumptions and Motivations } 

\subsection{Assumptions}
\label{sec:Assumptions}
Consider the following function:  for $y \in \mathbb{R}$
\begin{align}
&G_{\sigma_1,\sigma_2,\sfR,n}(y)\notag\\
&=\frac{\bbE\left[\frac{\sfR}{\|y+\bfW\|}\sfh_{\frac{n}{2}}\left(\frac{\sfR}{\sigma_2^2}\| y+\bfW\|\right)-1  \right]}{\sigma_2^2} -\frac{\frac{\sfR}{y}\sfh_{\frac{n}{2}}\left(\frac{\sfR}{\sigma_1^2}y\right) -1 }{\sigma_1^2}, 
	\end{align} 
	where $\bfW \sim  {\cal N}(\mathbf{0}_{n+2},(\sigma_2^2-\sigma_1^2)\bfI_{n+2})$. 
	
	In this work, in order to make progress on the secrecy-capacity, we make the following \emph{assumption} about the ratio of the Bessel functions:    for all $\sfR \ge 0, \sigma_2 \ge \sigma_1 \ge 0$ and $n \in \mathbb{N}$, the function  $y \mapsto G_{\sigma_1,\sigma_2,\sfR,n}(y)$ has \emph{at most} three sign changes. 
	
	In general, proving that $G_{\sigma_1,\sigma_2,\sfR,n}(y)$ has at most  three sign changes is not easy.  	However, extensive numerical evaluations show that this property holds for any $n, \sfR, \sigma_1, \sigma_2$. 
	
	 It is not hard to show that to the function $G_{\sigma_1,\sigma_2,\sfR,n}(y)$ is odd and there is a sign change at $y=0$.  Therefore, the problem boils down to showing that there is at most one sign change for $y>0$. 
Using this, we can give a sufficient condition for this assumption to be true. Note that
\begin{align}
G_{\sigma_1,\sigma_2,\sfR,n}(y)&\ge-\frac{1}{\sigma_2^2}+\frac{1}{\sigma_1^2}-\frac{\sfR}{\sigma_1^2 y}\sfh_{\frac{n}{2}}\left(\frac{\sfR}{\sigma_1^2}y\right) \label{eq:LB_on_h} \\
&\ge -\frac{1}{\sigma_2^2}+\frac{1}{\sigma_1^2}-\frac{\sfR^2}{\sigma_1^4 n}, \label{eq:UB_on_h}
\end{align}
which is nonnegative, hence has no sign change for $y>0$, for 
\begin{equation}
\sfR < \sigma_1^2 \sqrt{n \left(\frac{1}{\sigma_1^2}-\frac{1}{\sigma_2^2}\right)},
\end{equation}
for all $y\ge 0$. The inequality in~\eqref{eq:LB_on_h} follows by $\sfh_{\frac{n}{2}}(x)\ge 0$ for $x\ge 0$; and~\eqref{eq:UB_on_h} follows by $\sfh_{\frac{n}{2}}(x)\le \frac{x}{n}$ for $x\ge 0$ and $n\in \mathbb{N}$.

\subsection{Small Amplitude Regime} 

In this work a small amplitude regime is defined as follows. 
\begin{definition} Let $\bfX_{\sfR} \sim P_{\bfX_{\sfR}}$ be uniform on $\cC(\sfR)=\{ \bfx :  \|\bfx\|=\sfR \}$. The capacity in \eqref{eq:Secracy_CAP}  is said to be in the small amplitude regime if $\sfR \le \bar{\sfR}_n(\sigma_1^2,\sigma_2^2)$ where
\begin{equation}
\bar{\sfR}_n(\sigma_1^2,\sigma_2^2)= \max \left\{ \sfR:  P_{\bfX_{\sfR}} =\arg  \max_{\bfX \in  \cB_0(\sfR)}   I(\bfX; \bfY_1 | \bfY_2)  \right \}. \label{eq:small_amplitude_def}
\end{equation}
If the set in \eqref{eq:small_amplitude_def} is empty, we set $\bar{\sfR}_n(\sigma_1^2,\sigma_2^2)=0$. 
\end{definition}
The quantity $\bar{\sfR}_n(\sigma_1^2,\sigma_2^2)$ represents the largest radius $\sfR$ for which $P_{\bfX_{\sfR}}$ is secrecy-capacity-achieving. 

One of the main objectives of this work is to characterize $\bar{\sfR}_n(\sigma_1^2,\sigma_2^2)$. 
\subsection{Connections to Other Optimization Problems} 
\label{sec:Connection_Other_Problem}

The distribution $ P_{\bfX_{\sfR}}$ occurs in a variety of statistical and information-theoretic applications. For example, consider the following two optimization problems:
\begin{align}
 \max_{\bfX \in  \cB_0(\sfR)}&   I(\bfX; \bfX+\bfN),\\
  \max_{\bfX \in  \cB_0(\sfR)}& {\rm mmse}(\bfX| \bfX+\bfN),
 \end{align} 
where  $\bfN \sim \mathcal{N}(\mathbf{0}_n,\sigma^2 \bfI_n)$. The first problem seeks to characterize the capacity of the point-to-point channel under the amplitude constraint,
 and the second problem seeks to find the largest minimum mean squared error under the assumption that the signal has bounded amplitude; the interested reader is referred to \cite{dytsoMI_est_2019,favano2021capacity,berry1990minimax} for a detailed background on both problems.  
 
Similarly to the wiretap channel, we can define the small amplitude regime for both problems as the largest $\sfR$  such that $ P_{\bfX_{\sfR}}$ is optimal and denote these by $\bar{\sfR}_n^\text{ptp}(\sigma^2)$ and $\bar{\sfR}_n^\text{MMSE}(\sigma^2)$.
We now argue that both $\bar{\sfR}_n^\text{ptp}(\sigma^2)$ and $\bar{\sfR}_n^\text{MMSE}(\sigma^2)$
 can be seen as a special case of the wiretap solution. Hence, the wiretap channel provides and interesting unification and generalization of these two problems.  
 
 First,  note that the point-to-point solution can be recovered from the wiretap  by simply specializing the wiretap channel to the point-to-point channel, that is 
 \begin{align} \label{eq:Rptp}
 \bar{\sfR}_n^\text{ptp}(\sigma^2)= \lim_{\sigma_2 \to \infty} \bar{\sfR}_n(\sigma^2,\sigma_2^2).
 \end{align} 
 Second, to see that the MMSE solution can be recovered from the wiretap  recall that by the I-MMSE relationship \cite{I-MMSE},  we have that 
 \begin{align}
 & \max_{\bfX \in  \cB_0(\sfR) }  I(\bfX; \bfY_1) - I(\bfX; \bfY_2) \notag \\
  &=  \max_{\bfX \in  \cB_0(\sfR) }  \frac{1}{2} \int_{\sigma_1^2}^{\sigma_2^2} \frac{ {\rm mmse}(\bfX| \bfX+ \sqrt{s}\bfZ)}{s^2 } \rmd s
 \end{align} 
 where $\bfZ$ is standard Gaussian.  Now note that if we choose $\sigma_2^2=\sigma_1^2+\epsilon$ for some small enough $\epsilon>0$, we arrive at
 \begin{align}
 & \max_{\bfX \in  \cB_0(\sfR) }  I(\bfX; \bfY_1) - I(\bfX; \bfY_2) \\
 &=    \max_{\bfX \in  \cB_0(\sfR) }   \frac{\epsilon}{2}  \frac{ {\rm mmse}(\bfX| \bfX+ \sqrt{\sigma_1^2}\bfZ)}{\sigma_1^4 }.
 \end{align}  
 Consequently, for a small enough $\epsilon>0$, 
 \begin{equation}\label{eq:reduction_to_mmse}
 \bar{\sfR}_n^\text{MMSE}(\sigma^2)=  \bar{\sfR}_n(\sigma^2,\sigma^2+\epsilon).
 \end{equation} 

\section{Main Results} 

\subsection{Characterizing the Small Amplitude Regime} 

Our first main result characterizes the small amplitude regime.  
\begin{theorem}\label{thm:Char_Small_Amplitude}
Consider a function 
\begin{align}
& f(\sfR)\notag\\
&=\int_{\sigma_1^2}^{\sigma_2^2} \frac{\bbE \left[      \mathsf{h}_{\frac{n}{2}}^2\left(  \frac{\|  \sqrt{s}\bfZ\| \sfR}{s} \right) +     \mathsf{h}_{\frac{n}{2}}^2\left(  \frac{\|  \sfR+\sqrt{s}\bfZ\| \sfR}{s} \right) \right]-1}{s^2} \rmd s
\end{align} 
where $\bfZ \sim {\cal N}(\mathbf{0}_n,\bfI_n)$.
The input $\bfX_{\sfR}$ is secrecy-capacity achieving if  and only if $\sfR \le \bar{\sfR}_n(\sigma_1^2,\sigma_2^2)$ where $\bar{\sfR}_n(\sigma_1^2,\sigma_2^2)$ is given as the zero of 
\begin{equation}
f(\sfR)=0.  \label{eq:Condition_for_optimality}
\end{equation} 

\end{theorem} 

\begin{rem} Note that \eqref{eq:Condition_for_optimality} always has a solution. To see this observe that  $f(0)=\frac{1}{\sigma_2^2}-\frac{1}{\sigma_1^2}<0$, and $f(\infty)=\frac{1}{\sigma_1^2}-\frac{1}{\sigma_2^2}>0$. Moreover, the solution is unique, because $f(\sfR)$ is monotonically increasing for $\sfR\ge 0$.
\end{rem}

The solution to \eqref{eq:Condition_for_optimality} needs to be found numerically.\footnote{To avoid any loss of accuracy in the numerical evaluation of $\sfh_v(x)$ for large values of $x$, we used the exponential scaling provided in the MATLAB implementation of $\sfI_v(x)$.} Since evaluating $f(\sfR)$ is rather straightforward and not time-consuming, we opted for a binary search algorithm.
\begin{table}
\caption{Values of $\bar{\sfR}_n^{\text{ptp}}(1)$, $\bar{\sfR}_n ( 1,\sigma_2^2 )$, and $\bar{\sfR}_n^{\text{MMSE}}(1)$}
\def\arraystretch{1.25}%
\[\begin{tabular}{l  c c c c c c }
\toprule
  $n$                   		& 1     & 2     & 4     & 8     & 16    & 32     \\
\midrule
$\bar{\sfR}_n^{\text{ptp}}(1)$  & 1.666 & 2.454 & 3.580 & 5.158 & 7.367 & 10.472 \\
$\bar{\sfR}_n(1,1000)$  		& 1.664 & 2.450 & 3.575 & 5.151 & 7.357 & 10.458 \\    
$\bar{\sfR}_n(1,10)$    		& 1.518 & 2.221 & 3.229 & 4.646 & 6.632 & 9.424  \\
$\bar{\sfR}_n(1,1.5)$   		& 1.161 & 1.687 & 2.444 & 3.513 & 5.013 & 7.124  \\
$\bar{\sfR}_n(1,1.001)$ 		& 1.057 & 1.535 & 2.224 & 3.196 & 4.561 & 6.481  \\
$\bar{\sfR}_n^{\text{MMSE}}(1)$ & 1.057 & 1.535 & 2.223 & 3.195 & 4.560 & 6.479  \\

\end{tabular}\]
\label{Table1}
\vspace{-0.4cm}
\end{table}

 In Table~\ref{Table1}, we show the values of $\bar{\sfR}_n ( 1,\sigma_2^2 )$ for some values of~$\sigma_2^2$ and~$n$. Moreover, we report the values of $\bar{\sfR}_n^{\text{ptp}}(1)$ and $\bar{\sfR}_n^{\text{MMSE}}(1)$ from~\cite{dytsoMI_est_2019} in the first and the last row, respectively. As predicted by~\eqref{eq:Rptp}, we can appreciate the close match of the $\bar{\sfR}_n^{\text{ptp}}(1)$ row with the one of $\bar{\sfR}_n(1,1000)$. Similarly, the agreement between the $\bar{\sfR}_n^{\text{MMSE}}(1)$ row and the $\bar{\sfR}_n(1,1.001)$ row is justified by~\eqref{eq:reduction_to_mmse}.

\subsection{Large $n$ Asymptotics} 

We now use the result in Theorem~\ref{thm:Char_Small_Amplitude} to characterize the asymptotic behavior of $\bar{\sfR}_n(\sigma_1^2,\sigma_2^2)$. 
\begin{theorem}\label{thm:large_n_beh}
\begin{align} \label{eq:c_asym}
\lim_{n \to \infty} \frac{\bar{\sfR}_n(\sigma_1^2,\sigma_2^2)}{\sqrt{n}}=c(\sigma_1^2,\sigma_2^2),
\end{align}
where  $c(\sigma_1^2,\sigma_2^2)$ is the solution of 
\begin{equation}
 \int_{\sigma_1^2}^{\sigma_2^2} \frac{{ \frac{c^2 }{ \left(  \frac{\sqrt{s}}{2}+\sqrt{ \frac{s}{4} + c^2} \right)^2}} +      \frac{ c^2 (c^2+ s)}{ \left( \frac{s}{2}+\sqrt{ \frac{s^2}{4} +c^2( c^2+ s)  } \right)^2} -1}{s^2} \rmd s =0.
\end{equation} 
\end{theorem} 
\begin{IEEEproof}
See Section~\ref{sec:large_n_beh}. 
\end{IEEEproof}%
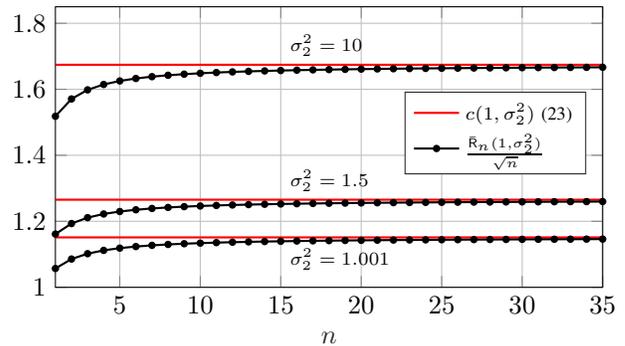
\begin{figure}[t]
	\centering
	\begin{tikzpicture}

\begin{axis}[%
width=\linewidth,
height=0.6\linewidth,
xmin=1,
xmax=35,
xlabel style={font=\color{white!15!black}},
xlabel={$n$},
ymin=1,
ymax=1.85,
ylabel style={font=\color{white!15!black}},
axis background/.style={fill=white},
xmajorgrids,
ymajorgrids,
legend style={legend cell align=left, align=left, draw=white!15!black,at={(0.97,0.55)},anchor=east}
]
\addplot [color=red, line width=0.8pt]
  table[row sep=crcr]{%
1	1.15125151780705\\
2	1.15125151780705\\
3	1.15125151780705\\
4	1.15125151780705\\
5	1.15125151780705\\
6	1.15125151780705\\
7	1.15125151780705\\
8	1.15125151780705\\
9	1.15125151780705\\
10	1.15125151780705\\
11	1.15125151780705\\
12	1.15125151780705\\
13	1.15125151780705\\
14	1.15125151780705\\
15	1.15125151780705\\
16	1.15125151780705\\
17	1.15125151780705\\
18	1.15125151780705\\
19	1.15125151780705\\
20	1.15125151780705\\
21	1.15125151780705\\
22	1.15125151780705\\
23	1.15125151780705\\
24	1.15125151780705\\
25	1.15125151780705\\
26	1.15125151780705\\
27	1.15125151780705\\
28	1.15125151780705\\
29	1.15125151780705\\
30	1.15125151780705\\
31	1.15125151780705\\
32	1.15125151780705\\
33	1.15125151780705\\
34	1.15125151780705\\
35	1.15125151780705\\
};
\addlegendentry{\scriptsize $c(1,\sigma_2^2) \ \eqref{eq:c_asym}$}

\addplot [color=black, line width=0.8pt, mark=*, mark size =1pt, mark options={solid, black}]
  table[row sep=crcr]{%
1	1.05700754961514\\
2	1.08567631436701\\
3	1.10184786913882\\
4	1.11190914815219\\
5	1.11867291122447\\
6	1.12349710272807\\
7	1.12709773157121\\
8	1.12988199389348\\
9	1.13209646405149\\
10	1.13389840141803\\
11	1.13539249136889\\
12	1.13665098589732\\
13	1.1377252939573\\
14	1.13865288306867\\
15	1.13946185044196\\
16	1.14017348091489\\
17	1.14080429167319\\
18	1.14136728158719\\
19	1.14187282962213\\
20	1.1423292467616\\
21	1.1427433827209\\
22	1.14312083565491\\
23	1.14346628037165\\
24	1.14378358654872\\
25	1.14407607154453\\
26	1.14434654435366\\
27	1.14459736408627\\
28	1.1448306274395\\
29	1.14504810200458\\
30	1.14525134228772\\
31	1.14544169089822\\
32	1.14562035265051\\
33	1.14578835697793\\
34	1.14594661770954\\
35	1.14609598552094\\
};
\addlegendentry{\scriptsize $\frac{\Bar{\mathsf{R}}_n(1,\sigma_2^2)}{\sqrt{n}}$}

\node[below right, align=left]
at (axis cs:15,1.147) {\scriptsize $\sigma_2^2 = 1.001$};
\addplot [color=red, line width=0.8pt, forget plot]
  table[row sep=crcr]{%
1	1.26546217419275\\
2	1.26546217419275\\
3	1.26546217419275\\
4	1.26546217419275\\
5	1.26546217419275\\
6	1.26546217419275\\
7	1.26546217419275\\
8	1.26546217419275\\
9	1.26546217419275\\
10	1.26546217419275\\
11	1.26546217419275\\
12	1.26546217419275\\
13	1.26546217419275\\
14	1.26546217419275\\
15	1.26546217419275\\
16	1.26546217419275\\
17	1.26546217419275\\
18	1.26546217419275\\
19	1.26546217419275\\
20	1.26546217419275\\
21	1.26546217419275\\
22	1.26546217419275\\
23	1.26546217419275\\
24	1.26546217419275\\
25	1.26546217419275\\
26	1.26546217419275\\
27	1.26546217419275\\
28	1.26546217419275\\
29	1.26546217419275\\
30	1.26546217419275\\
31	1.26546217419275\\
32	1.26546217419275\\
33	1.26546217419275\\
34	1.26546217419275\\
35	1.26546217419275\\
};
\addplot [color=black, line width=0.8pt, mark=*, mark size =1pt, mark options={solid, black}, forget plot]
  table[row sep=crcr]{%
1	1.16127361761653\\
2	1.19315741430016\\
3	1.21106514355155\\
4	1.22217627600009\\
5	1.2296336837499\\
6	1.23494726245077\\
7	1.23891067505508\\
8	1.2419741272445\\
9	1.24440997324439\\
10	1.24639161344061\\
11	1.24803444050257\\
12	1.24941806043064\\
13	1.25059902990421\\
14	1.25161869229161\\
15	1.2525078699636\\
16	1.25329001221689\\
17	1.25398330182204\\
18	1.25460204908811\\
19	1.25515761954344\\
20	1.25565921147753\\
21	1.2561143261815\\
22	1.25652909690446\\
23	1.25690868414368\\
24	1.25725736984404\\
25	1.25757874778852\\
26	1.25787596143686\\
27	1.25815156205166\\
28	1.25840788367657\\
29	1.25864684574435\\
30	1.25887015119169\\
31	1.25907929901797\\
32	1.2592755979502\\
33	1.2594601888743\\
34	1.25963408253852\\
35	1.2597982055379\\
};
\node[above right, align=left]
at (axis cs:15,1.262) {\scriptsize $\sigma_2^2 = 1.5$};
\addplot [color=red, line width=0.8pt, forget plot]
  table[row sep=crcr]{%
1	1.67419101387548\\
2	1.67419101387548\\
3	1.67419101387548\\
4	1.67419101387548\\
5	1.67419101387548\\
6	1.67419101387548\\
7	1.67419101387548\\
8	1.67419101387548\\
9	1.67419101387548\\
10	1.67419101387548\\
11	1.67419101387548\\
12	1.67419101387548\\
13	1.67419101387548\\
14	1.67419101387548\\
15	1.67419101387548\\
16	1.67419101387548\\
17	1.67419101387548\\
18	1.67419101387548\\
19	1.67419101387548\\
20	1.67419101387548\\
21	1.67419101387548\\
22	1.67419101387548\\
23	1.67419101387548\\
24	1.67419101387548\\
25	1.67419101387548\\
26	1.67419101387548\\
27	1.67419101387548\\
28	1.67419101387548\\
29	1.67419101387548\\
30	1.67419101387548\\
31	1.67419101387548\\
32	1.67419101387548\\
33	1.67419101387548\\
34	1.67419101387548\\
35	1.67419101387548\\
};
\addplot [color=black, line width=0.8pt, mark=*, mark size =1pt, mark options={solid, black}, forget plot]
  table[row sep=crcr]{%
1	1.51816337717656\\
2	1.57069402064057\\
3	1.59822577718413\\
4	1.61456083937775\\
5	1.62523074249726\\
6	1.63270567994323\\
7	1.63821954926654\\
8	1.64244899251058\\
9	1.64579347475647\\
10	1.64850318873694\\
11	1.65074252415981\\
12	1.6526238200227\\
13	1.654226388971\\
14	1.65560773234257\\
15	1.65681066364209\\
16	1.65786754577848\\
17	1.6588034487503\\
18	1.65963796851718\\
19	1.66038672797188\\
20	1.66106230884444\\
21	1.66167490812488\\
22	1.66223291677004\\
23	1.66274331915622\\
24	1.66321198947868\\
25	1.66364378974108\\
26	1.66404297491439\\
27	1.66441300778254\\
28	1.66475706820989\\
29	1.66507772436275\\
30	1.66537731120684\\
31	1.66565784212978\\
32	1.66592107114358\\
33	1.6661685346091\\
34	1.66640166278123\\
35	1.66662157904352\\
};
\node[above right, align=left]
at (axis cs:15,1.673) {\scriptsize $\sigma_2^2 = 10$};

\end{axis}

\end{tikzpicture}%
	\caption{Asymptotic behavior of $\Bar{\mathsf{R}}_n(1,\sigma_2^2)/\sqrt{n}$ versus $n$ for $\sigma_1^2 = 1$ and $\sigma_2^2 = 1.001,1.5,10$.}
	\label{fig:asymRn}
\end{figure}%
In Fig.~\ref{fig:asymRn}, for $\sigma_1^2 = 1$ and $\sigma_2^2 = 1.001,1.5,10$, we show the behavior of $\bar{\sfR}_n(1,\sigma_2^2)/\sqrt{n}$ and how its asymptotic converges to $c(1,\sigma_2^2)$.

\subsection{Capacity Expression in the Small Amplitude Regime} 
The result in Theorem~\ref{thm:Char_Small_Amplitude} can also be used to establish the secrecy-capacity for all $\sfR \le \bar{\sfR}_n(\sigma_1^2,\sigma_2^2)$ as is done next. 

\begin{theorem}\label{thm:Capacitiy_Small}  If $\sfR \le \bar{\sfR}_n(\sigma_1^2,\sigma_2^2)$, then 
\begin{equation} \label{eq:Cs}
C_s(\sigma_1^2, \sigma_2^2, \sfR)=  \frac{1}{2} \int_{\sigma_1^2}^{\sigma_2^2} \frac{\sfR^2 -\sfR^2\bbE \left[      \mathsf{h}_{\frac{n}{2}}^2\left(  \frac{\|  \sfR+\sqrt{s}\bfZ\| \sfR}{s} \right) \right] }{s^2} \rmd s.
\end{equation}
\end{theorem} 
\begin{IEEEproof}
See Section~\ref{sec:thm:Capacitiy_Small}. 
\end{IEEEproof}

\section{Beyond the Small Amplitude Regime} 
{ To evaluate the secrecy-capacity and find the optimal distribution $P_{\bfX^\star}$ beyond $\bar{\sfR}_n$ we rely on numerical estimations. We remark that, as pointed out in \cite{DytsoITWwiretap2018}, the capacity-achieving distribution is isotropic and consists of finitely many co-centric shells. Keeping this in mind, we can find the optimal input distribution $P_{\bfX^\star}$ by just optimizing over $P_{\|\bfX \|}$ with $\|\bfX \|\le \sfR$.
	
	Let us denote by $\widetilde{C}(\sigma_1^2,\sigma_2^2,\sfR)$ the numerical estimate of the secrecy-capacity and by $P_{\|\bfX^\star\|}$ the optimal pmf of the input norm. To numerically evaluate $\widetilde{C}(\sigma_1^2,\sigma_2^2,\sfR)$ and $P_{\|\bfX^\star\|}$ we rely on the algorithmic procedure described in~\cite{barletta2021numerical}. 
	    
	In Fig.~\ref{fig:Cs_vs_Cnum}, we show with black circles the numerical estimate $\widetilde{C}(\sigma_1^2,\sigma_2^2,\sfR)$ for $\sigma_1^2 = 1$, $\sigma_2^2 = 1.5, 10$, and $n=2,4$. For the same values of $\sigma_1^2$, $\sigma_2^2$, and $n$ we also show, with the red lines, the analytical small amplitude regime capacity $C_s(\sigma_1^2,\sigma_2^2,\sfR)$ from Theorem~\ref{thm:Capacitiy_Small}. Also, we show with blue dotted lines the secrecy-capacity under the average-power constraint $\bbE \left[ \| \bfX \|^2 \right] \leq \sfR^2$:
	\begin{align} \label{eq:C_G}
		C_G(\sigma_1^2,\sigma_2^2,\sfR) &= \frac{n}{2} \log \frac{1+\sfR^2/\sigma_1^2}{1+\sfR^2/\sigma_2^2}\ge C_s(\sigma_1^2,\sigma_2^2,\sfR),
	\end{align}
	where the inequality follows by noting that the average-power constraint $\bbE \left[ \| \bfX \|^2 \right] \leq \sfR^2$ is weaker than the amplitude constraint $\| \bfX \| \leq \sfR$. Finally, the dashed vertical lines show $\bar{\sfR}_n$ for the considered values of $\sigma_1^2$, $\sigma_2^2$, and $n$.
	
	In Fig.~\ref{fig:PMFevo_n2}, we show the evolution of the numerically estimated pmf $P_{\|\bfX^\star\|}$ for increasing values of $\sfR$, for $\sigma_1^2 = 1$, $\sigma_2^2 = 1.5$, and $n=2,8$. The figure shows, at each $\sfR$, the normalized amplitude mass points in the estimated pmf, while the size of the circles qualitatively shows the associated probability.}
\begin{figure}[t]
	\centering
	\input{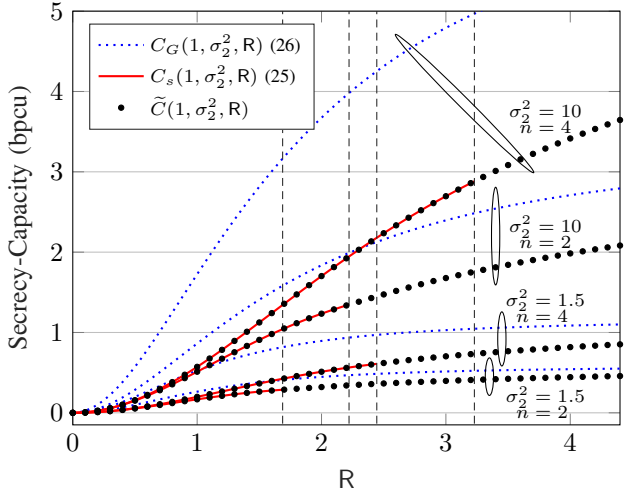}
	\caption{Secrecy-capacity in bit per channel use (bpcu) versus $\sfR$, for $\sigma_2^2 = 1.5,10$ and $n=2,4$.}\label{fig:Cs_vs_Cnum}
\end{figure}

\begin{figure}[t]
	\centering
	\input{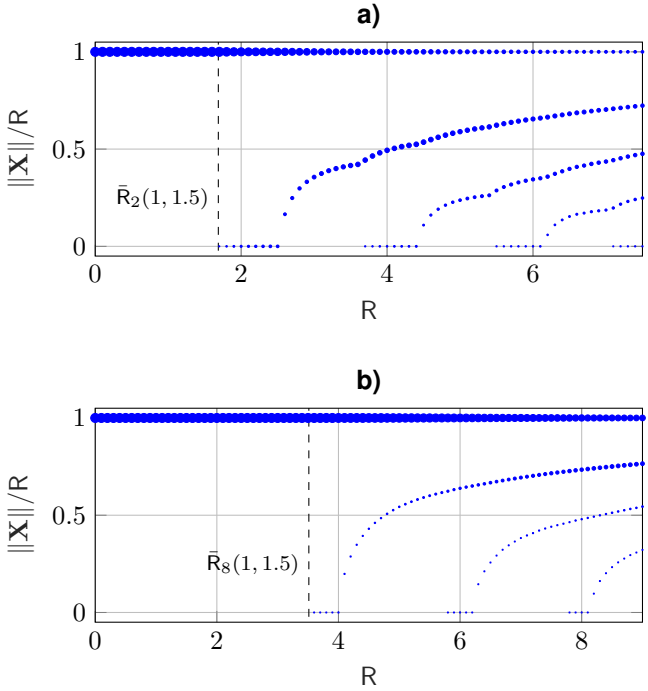}
	\caption{Evolution of the numerically estimated $P_{\|\bfX^\star\|}$ versus $\sfR$ for $\sigma_1^2 = 1$, $\sigma_2^2 = 1.5$, \textbf{\textsf{a)}} $n=2$, and \textbf{\textsf{b)}} $n=8$.}
	\label{fig:PMFevo_n2}
\end{figure}

\section{Proof of Theorem~\ref{thm:Char_Small_Amplitude}}

\subsection{KKT Conditions}
\begin{lemma}\label{lem:KKT} $P_{\bfX^\star}$ maximizes \eqref{eq:Secracy_CAP} if and only if 
\begin{align}
\Xi(\bfx;P_{\bfX^\star}) &= C_s(\sigma_1^2, \sigma_2^2,\sfR), \, \bfx \in \supp(P_{\bfX^\star}),  \label{eq:EqualityCOndition}\\
\Xi(\bfx;P_{\bfX^\star}) &\le C_s(\sigma_1^2, \sigma_2^2, \sfR), \, \bfx \in  \cB_0(\sfR) ,
\end{align}
where  for $\bfx \in \mathbb{R}^n$
\begin{align}
\Xi(\bfx;P_{\bfX^\star})&=\sfD(f_{\bfY_1|\bfX}(\cdot| \bfx) \|f_{\bfY_1^\star})- \sfD(f_{\bfY_2|\bfX}(\cdot|\bfx) \|f_{\bfY_2^\star})\\
&=\bbE \left[ g(\bfY_1) | \bfX=\bfx \right] , \label{eq:Writing_KKT_as_statistics}
\end{align}
and where 
\begin{align} \label{eq:functiong}
g(\bfy)=\bbE\left[\log\frac{f_{\bfY_2^\star}(\bfy+\bfN)}{f_{\bfY_1^\star}(\bfy)}\right]+ n \log\left(\frac{\sigma_2}{\sigma_1}\right),  \,  \bfy\in \mathbb{R}^n, 
\end{align} 
with $\bfN\sim {\cal N}(\mathbf{0}_n, (\sigma_2^2-\sigma_1^2) \bfI_n )$. 
\end{lemma} 
\begin{proof}
	This is a vector extension of \cite[Lemma~1]{barletta2021scalar}.
\end{proof}

\subsection{A New Necessary and Sufficient Condition}

\begin{theorem}\label{thm:equivalent_condition}  $P_{\bfX_{\sfR}} $ is optimal if and only if for all $\| \bfx \|=\sfR$
\begin{equation}\label{eq:ineq_density}
\Xi({\bf 0};P_{\bfX_{\sfR}}) \le \Xi(\bfx;P_{\bfX_{\sfR}}). 
\end{equation} 
Moreover, if 
\begin{equation}
\sfR < \sigma_1^2 \sqrt{n \left(\frac{1}{\sigma_1^2}-\frac{1}{\sigma_2^2}\right)}
\end{equation}
then $P_{\bfX_{\sfR}} $ is optimal.
\end{theorem} 

\begin{proof}
	The secrecy-density $\Xi(\cdot; P_{\bfX_{\sfR}})$ is a function only of $\| \bfx\|$, thanks to the rotational symmetry of the Gaussian distribution and of $P_{\bfX_{\sfR}}$. In view of this, a way to prove condition~\eqref{eq:ineq_density} is to show that the maximum of $\|\bfx\|\mapsto\Xi(\|\bfx\|;P_{\bfX_{\sfR}})$ occurs at either $\|\bfx\|=0$ or $\|\bfx\|=\sfR$. Next, we show that the derivative of $\Xi(\|\bfx\|;P_{\bfX_{\sfR}})$ makes at most one sign change, from negative to positive. This fact will prove the claim.
	
	From Lemma \ref{Lemma:derivative_Xi} in the Appendix, the derivative of $\Xi$ is\footnote{A related calculation was erroneously performed in \cite{dytsoMI_est_2019}. However, this error does not change the results of \cite{dytsoMI_est_2019} as only the sign of the derivative is important and not the value itself.} 
	\begin{align} \label{eq:derivative_Xi}
	\Xi'(\|\bfx\|;P_{\bfX_{\sfR}}) =\|\bfx \|\:\bbE\left[\widetilde{M}_2(\sigma_1 Q_{n+2})-M_1(\sigma_1 Q_{n+2}) \right]
	\end{align}
	where $Q_{n+2}^2$ is a noncentral chi-square random variable with $n+2$ degrees of freedom and noncentrality parameter $\frac{\|\bfx\|^2}{\sigma_1^2}$ and
	\begin{align}
	M_i(y) &= \frac{1}{\sigma_i^2}\left(\frac{\sfR}{y}\sfh_{  \frac{n}{2} }\left(\frac{\sfR}{\sigma_i^2}y\right)-1\right), \qquad i\in \{1,2\} \\
	\widetilde{M}_2(y) &= \bbE\left[M_2(\|y+\bfW \|)\right],
	\end{align}
	where  $\bfW \sim {\cal N}(\mathbf{0}_{n+2},(\sigma_2^2-\sigma_1^2)\bfI_{n+2})$.

Note that $\Xi'(0;P_{\bfX_{\sfR}})=0$, and that $\Xi'(\|\bfx\|;P_{\bfX_{\sfR}})>0$ for sufficiently large $\|\bfx\|$; in fact, we have
\begin{align}
&\Xi'(\|\bfx\|;P_{\bfX_{\sfR}}) > \|\bfx\|\left(\frac{1}{\sigma_1^2}-\frac{1}{\sigma_2^2}\right) -\frac{\|\bfx\|}{\sigma_1^2}\bbE\left[\frac{\sfR}{\sigma_1 Q_{n+2}}  \right] \label{eq:applyboundsonh} \\
&=\|\bfx\|\left(\frac{1}{\sigma_1^2}-\frac{1}{\sigma_2^2}\right) -\frac{\|\bfx\|}{\sigma_1^2}\bbE\left[\frac{\sfR}{\|\bfx\|} \sfh_{\frac{n}{2}}\left(\frac{\|\bfx\|}{\sigma_1}Q_n\right) \right] \label{eq:ndegreesoffreedom} \\
&\ge\|\bfx\|\left(\frac{1}{\sigma_1^2}-\frac{1}{\sigma_2^2}\right) -\frac{\sfR}{\sigma_1^2}, \label{eq:applyboundsonh1}
\end{align}  
where \eqref{eq:applyboundsonh} follows by $0\le \sfh_{\frac{n}{2}}(x) \le 1$ for $x\ge 0$; \eqref{eq:ndegreesoffreedom} follows by a change of measure in the expectation; and finally \eqref{eq:applyboundsonh1} holds by $\sfh_{\frac{n}{2}}(x) \le 1$.

 To conclude, we need to prove that $\Xi'(\|{\bfx} \|;P_{\bfX_{\sfR}})$ changes sign at most once.  To that end, we will need the following lemma  shown in \cite[Theorem 3]{karlin1957polya}. 
\begin{lemma} \label{lem:numberOfOscillations} Let the pdf $f(x,\omega)$ be a positive-definite kernel that can be differentiated $n$ times with respect to $x$ for all $\omega$, and let $\eta(\omega)$ be a function that changes  sign $n$ times. If
\begin{align}
M(x)= \int    \eta(\omega)     f(x,\omega) {\rm d} \omega,
\end{align} 
can be differentiated $n$ times, then $M(x)$ changes sign at most $n$ times. 
\end{lemma}

By using~\eqref{eq:derivative_Xi}, the fact that the pdf of a chi-square is positive defined kernel~\cite{karlin1957polya}, and Lemma~\ref{lem:numberOfOscillations}, the number of sign changes of $\Xi'(\|{\bfx} \|;P_{\bfX_{\sfR}})$ is upper-bounded by the number of sign changes of 
\begin{equation}
\widetilde{M}_2(y)-M_1(y)=G_{\sigma_1,\sigma_2,\sfR,n}(y), 
\end{equation}
for $y>0$ where $G_{\sigma_1,\sigma_2,\sfR,n}(y)$ was defined and discussed in Section~\ref{sec:Assumptions} and it was assumed that it has at most one sign change for $y >0$.  For example, a sufficient condition is given by 
\begin{equation}
\sfR < \sigma_1^2 \sqrt{n \left(\frac{1}{\sigma_1^2}-\frac{1}{\sigma_2^2}\right)}
\end{equation}
  This concludes the proof.

\end{proof}

\subsection{Estimation Theoretic Representation}

To complete the proof we seek to re-write the condition in Theorem~\ref{thm:equivalent_condition} in the estimation theoretic form.  To that end, we need the following representation of the relative entropy  \cite{verdu2010mismatched}: 
\begin{equation}
D(P_{\bfX_1+\sqrt{t}\bfZ} \| P_{\bfX_2+\sqrt{t}\bfZ})
= \frac{1}{2} \int_t^\infty  \frac{g(s)}{s^2} \rmd s , \label{eq:mistmatchedMMSE_formulat}
\end{equation} 
where 
\begin{align}
g(s)&= \bbE \left[ \| \bfX_1 -\phi_2(\bfX_1+\sqrt{s} \bfZ)    \|^2 \right] \notag\\
&\quad - \bbE \left[ \| \bfX_1 -\phi_1(\bfX_1+\sqrt{s} \bfZ)    \|^2 \right]
\end{align} 
and where
\begin{align}
\phi_i(\bfy)=\bbE[\bfX_i|\bfX_i+\sqrt{s} \bfZ=\bfy], \,  i\in \{1,2\}. 
\end{align} 

Another fact that will be important for our expression is 
\begin{align}
\bbE \left[\bfX_{\sfR} \mid \bfX_{\sfR}+\sqrt{s} \bfZ=\bfy \right]=   \frac{\sfR \bfy}{\|\bfy\|}    \mathsf{h}_{\frac{n}{2}}\left(  \frac{\| \bfy\| \sfR}{s} \right),  \label{eq:ConditionalExpectation}
\end{align} 
see, for example \cite{dytsoMI_est_2019}, for the proof. 

Next,  using \eqref{eq:mistmatchedMMSE_formulat} and \eqref{eq:ConditionalExpectation} note that for any $\| \bfx\|=\sfR$ we have that for $i \in \{1,2\}$
\begin{align}
& D(P_{ {\bfx}+\sqrt{\sigma^2_i}\bfZ} \| P_{\bfX_{\sfR}+\sqrt{\sigma^2_i}\bfZ})\\
 &=\frac{1}{2} \int_{\sigma_i^2}^\infty  \frac{ \bbE \left[  \| \bfx -  \frac{\sfR ( {\bfx}+\sqrt{s}\bfZ)}{\| {\bfx}+\sqrt{s}\bfZ\|}    \mathsf{h}_{\frac{n}{2}}\left(  \frac{\|  {\bfx}+\sqrt{s}\bfZ\| \sfR}{s} \right) \|^2\right] }{s^2} \rmd s\\
  &=\frac{1}{2} \int_{\sigma_i^2}^\infty  \frac{\sfR^2 -\sfR^2\bbE \left[      \mathsf{h}_{\frac{n}{2}}^2\left(  \frac{\|  {\bfx}+\sqrt{s}\bfZ\| \sfR}{s} \right)\right] }{s^2} \rmd s, \label{eq:KL_at_x}
\end{align}
and
\begin{align}
D(P_{ {\bf 0}+\sqrt{\sigma^2_i}\bfZ} \| P_{\bfX_{\sfR}+\sqrt{\sigma^2_i}\bfZ})=\frac{1}{2} \int_{\sigma_i^2}^\infty  \frac{ \sfR^2 \bbE \left[ \sfh_{  \frac{n}{2} }^2 \left( \frac{ \sfR \|\bfZ\|}{s} \right)\right] }{s^2} \rmd s. \label{eq:KL_at_0}
\end{align}

Now, note that by using definition of  $\Xi(\bfx; P_{\bfX_{\sfR}})$ in  \eqref{eq:Writing_KKT_as_statistics}, and \eqref{eq:KL_at_x} and \eqref{eq:KL_at_0} we have that  for $\| \bfx\|=\sfR$
\begin{align}
&\Xi(\bfx; P_{\bfX_{\sfR}})\notag\\
&=D(P_{ {\bfx}+\sqrt{\sigma^2_1}\bfZ} \| P_{\bfX_{\sfR}+\sqrt{\sigma^2_1}\bfZ})-  D(P_{ {\bf x}+\sqrt{\sigma^2_2}\bfZ} \| P_{\bfX_{\sfR}+\sqrt{\sigma^2_2}\bfZ})\\
&=\frac{1}{2} \int_{\sigma_1^2}^{\sigma_2^2} \frac{\sfR^2 -\sfR^2\bbE \left[      \mathsf{h}_{\frac{n}{2}}^2\left(  \frac{\|  {\bfx}+\sqrt{s}\bfZ\| \sfR}{s} \right) \right] }{s^2} \rmd s, \label{eq:diff_of_KL_s}
\end{align} 
and 
\begin{align}
&\Xi({\bf 0} ;P_{\bfX_{\sfR}})\notag\\
&=D(P_{ {\bf 0}+\sqrt{\sigma^2_1}\bfZ} \| P_{\bfX_{\sfR}+\sqrt{\sigma^2_1}\bfZ})-  D(P_{ {\bf 0}+\sqrt{\sigma^2_2}\bfZ} \| P_{\bfX_{\sfR}+\sqrt{\sigma^2_2}\bfZ})\\
&=\frac{1}{2} \int_{\sigma_1^2}^{\sigma_2^2} \frac{\sfR^2\bbE \left[      \mathsf{h}_{\frac{n}{2}}^2\left(  \frac{\|  \sqrt{s}\bfZ\| \sfR}{s} \right) \right] }{s^2} \rmd s
\end{align}

Consequently, the necessary and sufficient condition in Theorem~\ref{thm:equivalent_condition} can be equivalently written as
\begin{align}
& \int_{\sigma_1^2}^{\sigma_2^2} \frac{\bbE \left[      \mathsf{h}_{\frac{n}{2}}^2\left(  \frac{\|  \sqrt{s}\bfZ\| \sfR}{s} \right) +     \mathsf{h}_{\frac{n}{2}}^2\left(  \frac{\|  {\bfx}+\sqrt{s}\bfZ\| \sfR}{s} \right) \right]-1}{s^2} \rmd s \le 0.\label{eq:Final_inequality_thm1} 
\end{align} 

Now $\bar{\sfR}_n(\sigma_1^2,\sigma_2^2)$ will be the largest $\sfR$ that satisfies \eqref{eq:Final_inequality_thm1}, which concludes the proof of Theorem~\ref{thm:Char_Small_Amplitude}.

\section{Proof of Theorem~\ref{thm:large_n_beh}} 
\label{sec:large_n_beh}
The objective of the proof is to understand how the condition in \eqref{eq:Condition_for_optimality} behaves as $n \to \infty$. To study the large $n$ behavior we  will need to the following bounds on the $ \mathsf{h}_{\nu}$ \cite{segura2011bounds,baricz2015bounds}:  for $\nu > \frac{1}{2}$ 
\begin{align}
 \mathsf{h}_{\nu}(x)=   \frac{x}{ \frac{2\nu-1}{2}+\sqrt{ \frac{(2\nu-1)^2}{4} +x^2}} \cdot g_\nu(x),
\end{align}
where 
\begin{align}
1 \ge g_\nu(x) \ge    \frac{ \frac{2\nu-1}{2}+\sqrt{ \frac{(2\nu-1)^2}{4} +x^2}}{ \nu+\sqrt{ \nu^2 +x^2} }.
\end{align} 

Now let $\sfR= c\sqrt{n}$ for some $c>0$.  The goal is to understand the behavior of $\bbE \left[      \mathsf{h}_{\frac{n}{2}}^2\left(  \frac{\|  \sqrt{s}\bfZ\| \sfR}{s} \right) +     \mathsf{h}_{\frac{n}{2}}^2\left(  \frac{\|  {\bfx}+\sqrt{s}\bfZ\| \sfR}{s} \right) \right]$ as $n$ goes to infinity. 
First,   let 
\begin{align}
V_n= \frac{\| \bfZ\|}{\sqrt{n}},
\end{align}
and note that 
\begin{align}
&\lim_{n \to \infty} \bbE \left[      \mathsf{h}_{\frac{n}{2}}^2\left(  \frac{\|  \sqrt{s}\bfZ\| c \sqrt{n}}{s} \right)  \right]\notag\\
&= \lim_{n \to \infty}  \bbE \left[   \left(  \frac{ \frac{ c V_n }{\sqrt{s}}}{ \frac{n-1}{2n}+\sqrt{ \frac{(n-1)^2}{4n^2} + \left(\frac{ c V_n }{\sqrt{s}} \right)^2}} \cdot g_{ \frac{n}{2}}  \left( \frac{ c V_n }{\sqrt{s}} n \right)\right)^2 \right]\\
&= \bbE \left[    \lim_{n \to \infty}   \left(  \frac{ \frac{ c V_n }{\sqrt{s}}}{ \frac{n-1}{2n}+\sqrt{ \frac{(n-1)^2}{4n^2} + \left(\frac{ c V_n }{\sqrt{s}} \right)^2}} \cdot g_{ \frac{n}{2}}  \left( \frac{ c V_n }{\sqrt{s}} n \right)\right)^2 \right] \label{eq:Applying_DCT} \\
 &= \frac{c^2 }{ \left(  \frac{\sqrt{s}}{2}+\sqrt{ \frac{s}{4} + c^2} \right)^2},\label{eq:using_Law_large_numbers}
\end{align} 
where \eqref{eq:Applying_DCT} follows by the dominated convergence theorem, and \eqref{eq:using_Law_large_numbers} follows since by the law of large numbers we have, almost surely, that 
\begin{align}
\lim_{n \to \infty} V_n^2= \lim_{n \to \infty} \frac{1}{n} \sum_{i=1}^n Z_i^2 =\bbE[Z^2]=1.
\end{align}

Second, let 
\begin{align}
W_n= \frac{\|  {\bfx}+\sqrt{s}\bfZ\|}{\sqrt{n}},
\end{align}
where without loss of generality we take $\bfx=[ \sfR, 0, \ldots, 0]$

\begin{align}
& \lim_{n \to \infty} \bbE \left[         \mathsf{h}_{\frac{n}{2}}^2\left(  \frac{ \|  {\bfx}+\sqrt{s}\bfZ\| c \sqrt{n}}{s} \right) \right] \notag\\
&= \lim_{n \to \infty}  \bbE \left[   \left(  \frac{ \frac{ c W_n }{s}}{ \frac{n-1}{2n}+\sqrt{ \frac{(n-1)^2}{4n^2} + \left(\frac{ c W_n }{s} \right)^2}} \cdot g_{ \frac{n}{2}}  \left( \frac{ cW_n }{s} n \right)\right)^2 \right]\\
&=   \bbE \left[ \lim_{n \to \infty}  \left(  \frac{ \frac{ c W_n }{s}}{ \frac{n-1}{2n}+\sqrt{ \frac{(n-1)^2}{4n^2} + \left(\frac{ c W_n }{s} \right)^2}} \cdot g_{ \frac{n}{2}}  \left( \frac{ cW_n }{s} n \right)\right)^2 \right]  \label{eq:Applying_DCT_v2} \\
&=    \frac{ c^2 (c^2+ s)}{ \left( \frac{s}{2}+\sqrt{ \frac{s^2}{4} +c^2( c^2+ s)  } \right)^2} , \label{eq:using_Law_large_numbers_second}
\end{align} 
where \eqref{eq:Applying_DCT_v2} follows by the dominated convergence theorem and 
where \eqref{eq:using_Law_large_numbers_second} follows since by the strong law of large numbers we have that almost surely
\begin{align}
\lim_{n \to \infty} W_n^2&= \lim_{n \to \infty} \frac{1}{n}   ( \sqrt{s} Z_1+c\sqrt{n})^2 + s \lim_{n \to \infty}  \frac{1}{n} \sum_{i=2}^nZ_i^2\\
 &= c^2+ s.
\end{align}

Combining \eqref{eq:using_Law_large_numbers} and \eqref{eq:using_Law_large_numbers_second} with \eqref{eq:Condition_for_optimality} we arrive at  
\begin{align}
& \int_{\sigma_1^2}^{\sigma_2^2} \frac{{ \frac{c^2 }{ \left(  \frac{\sqrt{s}}{2}+\sqrt{ \frac{s}{4} + c^2} \right)^2}} +      \frac{ c^2 (c^2+ s)}{ \left( \frac{s}{2}+\sqrt{ \frac{s^2}{4} +c^2( c^2+ s)  } \right)^2} -1}{s^2} \rmd s =0.
\end{align}

\section{Proof of Theorem~\ref{thm:Capacitiy_Small}} 
\label{sec:thm:Capacitiy_Small}
Using the KKT conditions in \eqref{eq:EqualityCOndition}, we have that for $\bfx=[\sfR, 0, \ldots, 0]$ 
\begin{align}
C_s(\sigma_1^2, \sigma_2^2, \sfR)&=\Xi(\bfx;P_{\bfX_{\sfR}})\\
&=\sfD(f_{\bfY_1|\bfX}(\cdot| \bfx) \|f_{\bfY_1^\star})- \sfD(f_{\bfY_2|\bfX}(\cdot|\bfx) \|f_{\bfY_2^\star})\\
&=  \frac{1}{2} \int_{\sigma_1^2}^{\sigma_2^2} \frac{\sfR^2 -\sfR^2\bbE \left[      \mathsf{h}_{\frac{n}{2}}^2\left(  \frac{\|  \sfR+\sqrt{s}\bfZ\| \sfR}{s} \right) \right] }{s^2} \rmd s
\end{align}
where the last expression was computed in \eqref{eq:diff_of_KL_s}.  This concludes the proof. 

\section{Conclusion} 
This paper focuses on the secrecy-capacity vector Gaussian wiretap channel under the peak-power (or amplitude constraint) in a so-called small (but not vanishing) amplitude regime. In this regime, the optimal input distribution $P_{\bfX_\sfR}$ is supported on a single sphere of radius $\mathsf{R}$.  The paper has identified the largest $\bar{\mathsf{R}}_n$ such that this distribution $P_{\bfX_\sfR}$ is optimal. In addition, the asymptotic of $\bar{\mathsf{R}}_n$ has been completely characterized as dimension $n$ approaches infinity. As a by-product of the analysis, the capacity in the small-amplitude regime has also been characterized in more or less closed-form.  The paper has also provided a number of supporting numerical examples. As part of ongoing work, we are trying to resolve the conjecture that was made regarding the number of zeros of the function defined through the ratios of Bessel functions. An interesting and ambitious future direction would be to determine a regime in which a mixture of a mass point at zero and $P_{\bfX_\sfR}$ is optimal.  

\appendix
\section{Derivative of the Secrecy-Density}
\begin{lemma} \label{Lemma:derivative_Xi}
	The derivative of the secrecy-density for the input $P_{\bfX_{\sfR}}$ is
	\begin{align} \label{eq:derivative_Xi_lemma}
	\Xi'(\|\bfx\|;P_{\bfX_{\sfR}}) =\|\bfx \|\:\bbE\left[\widetilde{M}_2(\sigma_1 Q_{n+2})-M_1(\sigma_1 Q_{n+2}) \right]
	\end{align}
	where $Q_{n+2}^2$ is a noncentral chi-square random variable with $n+2$ degrees of freedom and noncentrality parameter $\frac{\|\bfx\|^2}{\sigma_1^2}$ and
	\begin{align}
	M_i(y) &= \frac{1}{\sigma_i^2}\left(\frac{\sfR}{y}\sfh_{  \frac{n}{2} }\left(\frac{\sfR}{\sigma_i^2}y\right)-1\right), \qquad i\in \{1,2\} \\
	\widetilde{M}_2(y) &= \bbE\left[M_2(\|y+\bfW \|)\right],
	\end{align}
	where  $\bfW \sim {\cal N}(\mathbf{0}_{n+2},(\sigma_2^2-\sigma_1^2)\bfI_{n+2})$.
\end{lemma}
\begin{proof}
	We start with the secrecy-density expressed in spherical coordinates. A quick way to get the information densities in this coordinate system is to note that:
	\begin{align}
	&I(\bfX;\bfY_i) \nonumber\\
	 &= h(\bfY_i)-h(\bfN_i) \\
	&=h(\|\bfY_i \|)+(n-1)\bbE[\log\|\bfY_i \|]+h_\lambda\left(\frac{\bfY_i}{\|\bfY_i \|}  \right)- h(\bfN_i) \label{eq:spherical} \\
	&=h(\|\bfY_i \|^2)+\left(\frac{n}{2}-1\right)\bbE[\log\|\bfY_i \|^2] \nonumber\\
	&\quad+\log\frac{\pi^{\frac{n}{2}}}{\Gamma\left(\frac{n}{2}\right)}- \frac{n}{2}\log(2\pi e \sigma_i^2) \label{eq:square} \\
	&=h\left(\sigma_i^2 \left\|\frac{\bfX}{\sigma_i}+\widetilde{\bfN}_i  \right\|^2\right) \nonumber\\
	&\quad+\left(\frac{n}{2}-1\right)\bbE\left[\log\left(\sigma_i^2\left\|\frac{\bfX}{\sigma_i}+\widetilde{\bfN}_i  \right\|^2\right)\right] \nonumber\\
	&\quad+\log\frac{\pi^{\frac{n}{2}}}{\Gamma\left(\frac{n}{2}\right)}- \frac{n}{2}\log(2\pi e \sigma_i^2) \label{eq:chi_square_introduction} \\
	&=h\left( \left\|\frac{\bfX}{\sigma_i}+\widetilde{\bfN}_i  \right\|^2\right)+\left(\frac{n}{2}-1\right)\bbE\left[\log\left\|\frac{\bfX}{\sigma_i}+\widetilde{\bfN}_i  \right\|^2\right] \nonumber\\
	&\quad- \log\left((2 e)^\frac{n}{2}  \Gamma\left(\frac{n}{2}\right)\right)
	\end{align}
	where~\eqref{eq:spherical} holds by~\cite[Lemma~6.17]{lapidoth2003capacity} and by independence between $\|\bfY_i \|$ and $\frac{\bfY_i}{\|\bfY_i \|}$; the term $h_\lambda(\cdot)$ is a differential entropy-like quantity for random vectors on the $n$-dimensional unit sphere  \cite[Lemma~6.16]{lapidoth2003capacity}; \eqref{eq:square} holds because $\frac{\bfY_i}{\|\bfY_i \|}$ is uniform on the unit sphere and thanks to~\cite[Lemma~6.15]{lapidoth2003capacity}; and in~\eqref{eq:chi_square_introduction} we have $\widetilde{\bfN}_i\sim {\cal N}(\mathbf{0}_n,\bfI_n)$. It is now immediate to write the secrecy-density as follows:
	\begin{equation}
	\Xi(\|\bfx \|;P_{\bfX}) = i_1(\|\bfx \|;P_{\bfX})-i_2(\|\bfx \|;P_{\bfX})
	\end{equation}
	where
	\begin{align}
	&i_j(\|\bfx \|;P_{\bfX}) \nonumber\\
	&= -\int_0^{\infty} f_{\chi^2_{n}(\frac{\|\bfx \|^2}{\sigma_j^2})}(y) \log \frac{\int_0^\sfR f_{\chi^2_{n}(\frac{t^2}{\sigma_j^2})}(y)dP_{\|\bfX\|}(t)}{y^{\frac{n}{2}-1}}dy \nonumber\\
	&\quad - \log\left((2e)^{\frac{n}{2}} \Gamma\left( \frac{n}{2} \right) \right),
	\end{align}
	for $j\in\{1,2\}$.

	Given two values $\rho_1,\rho_2$ with $\rho_1>\rho_2$, write
	\begin{align}
	&i_j(\rho_1;P_{\bfX})-i_j(\rho_2;P_{\bfX}) \nonumber\\
	 &= \int_{0}^{\infty} \left(f_{\chi^2_{n}(\frac{\rho_1^2}{\sigma_j^2})}(y)-f_{\chi^2_{n}(\frac{\rho_2^2}{\sigma_j^2})}(y)\right)\log\frac{y^{\frac{n}{2}-1}}{f_{\|\frac{\bfY}{\sigma_j}\|^2}(y;P_{\bfX})} dy \\
	&=\int_{0}^{\infty} \left(F_{\chi^2_{n}(\frac{\rho_2^2}{\sigma_j^2})}(y)-F_{\chi^2_{n}(\frac{\rho_1^2}{\sigma_j^2})}(y)\right)\frac{d}{dy}\log\frac{y^{\frac{n}{2}-1}}{f_{\|\frac{\bfY}{\sigma_j}\|^2}(y;P_{\bfX})} dy \label{eq:der1}
	\end{align}
	where we have integrated by parts. Now notice that
	\begin{equation}
	\int_{0}^{\infty} \left(F_{\chi^2_{n}(\frac{\rho_2^2}{\sigma_j^2})}(y)-F_{\chi^2_{n}(\frac{\rho_1^2}{\sigma_j^2})}(y)\right) dy = \frac{\rho_1^2 - \rho_2^2}{\sigma_j^2}. \label{eq:auxpdf}
	\end{equation}
	Since $\chi^2_{n}(\frac{\rho_1^2}{\sigma_j^2})$ statistically dominates $\chi^2_{n}(\frac{\rho_2^2}{\sigma_j^2})$, the integrand function in \eqref{eq:auxpdf} is always positive. We can introduce an auxiliary output random variable $Q_j$, for $j\in\{1,2\}$, with pdf
	\begin{equation}\label{eq:def_fQ}
	f_{Q_j}(y;\rho_1,\rho_2) =\frac{\sigma_j^2}{\rho_1^2 - \rho_2^2} \left(F_{\chi^2_{n}(\frac{\rho_2^2}{\sigma_j^2})}(y)-F_{\chi^2_{n}(\frac{\rho_1^2}{\sigma_j^2})}(y)\right), 
	\end{equation} 
	for $y>0$,	to rewrite \eqref{eq:der1} as follows:
	\begin{align}
	&i_j(\rho_1;P_{\bfX})-i_j(\rho_2;P_{\bfX}) \nonumber\\
	 &= -\frac{\rho_1^2-\rho_2^2}{\sigma_j^2}\int_{0}^{\infty} f_{Q_j}(y;\rho_1,\rho_2)\frac{d}{dy}\log\frac{f_{\|\frac{\bfY}{\sigma_j}\|^2}(y;P_{\bfX})}{y^{\frac{n}{2}-1}} dy.\label{eq:afterQ}
	\end{align}
	We evaluate the derivative in \eqref{eq:afterQ} as:
	\begin{align}
	&\frac{d}{dy}\log\frac{f_{\|\frac{\bfY}{\sigma_j}\|^2}(y;P_{\bfX})}{y^{\frac{n}{2}-1}} \nonumber\\
	&= \frac{y^{\frac{n}{2}-1}}{f_{\|\frac{\bfY}{\sigma_j}\|^2}(y;P_{\bfX})}\int_{0}^{\sfR} \frac{d}{dy} \frac{f_{\chi^2_{n}(\frac{t^2}{\sigma_j^2})}(y)}{y^{\frac{n}{2}-1}} dP_{\|\bfX\|}(t) \label{eq:fy} \\
	&=\frac{y^{\frac{n}{2}-1}}{f_{\|\frac{\bfY}{\sigma_j}\|^2}(y;P_{\bfX})}\nonumber\\
	&\quad\int_{0}^{\sfR} \left(\frac{f_{\chi^2_{n-2}(\frac{t^2}{\sigma_j^2})}(y)}{2y^{\frac{n}{2}-1}}-\left(\frac{1}{2}+\frac{\frac{n}{2}-1}{y}\right)\frac{f_{\chi^2_{n}(\frac{t^2}{\sigma_j^2})}(y)}{y^{\frac{n}{2}-1}}\right) dP_{\|\bfX\|}(t) \label{eq:der_chisquare} \\
	&=\bbE\left[\frac{1}{2} \frac{f_{\chi^2_{n-2}(\frac{\|\bfX\|^2}{\sigma_j^2})}(\frac{\left\|\bfY\right\|^2}{\sigma_j^2})}{f_{\chi^2_{n}(\frac{\|\bfX\|^2}{\sigma_j^2})}(\frac{\left\|\bfY\right\|^2}{\sigma_j^2})}-\left(\frac{1}{2}+\frac{\frac{n}{2}-1}{\frac{\left\|\bfY\right\|^2}{\sigma_j^2}}\right) | \frac{\left\|\bfY\right\|^2}{\sigma_j^2}=y\right] \\
	&=\bbE\left[\frac{1}{2}\frac{\|\bfX\|}{\|\bfY\|} \frac{\sfI_{\frac{n}{2}-2}(\frac{\|\bfX\| \|\bfY\|}{\sigma_j^2})}{\sfI_{\frac{n}{2}-1}(\frac{\|\bfX\| \|\bfY\|}{\sigma_j^2})}-\left(\frac{1}{2}+\frac{\frac{n}{2}-1}{\frac{\left\|\bfY\right\|^2}{\sigma_j^2}}\right) | \frac{\left\|\bfY\right\|^2}{\sigma_j^2}=y\right] \label{eq:derlog} \\
	&=\bbE\left[\frac{1}{2}\frac{\|\bfX\|}{\|\bfY\|}\sfh_{\frac{n}{2}}\left(\frac{\|\bfX\| \|\bfY\|}{\sigma_j^2}\right)-\frac{1}{2} |   \frac{\left\|\bfY\right\|^2}{\sigma_j^2}=y\right] \label{eq:recurrence_relation}
	\end{align} 
	where in \eqref{eq:fy} we used
	\begin{equation}
	f_{\|\frac{\bfY}{\sigma_j}\|^2}(y;P_{\bfX}) = \int_{0}^{\sfR} f_{\chi^2_{n}(\frac{t^2}{\sigma_j^2})}(y) dP_{\|\bfX\|}(t);
	\end{equation} in \eqref{eq:der_chisquare} we used the relationship
	\begin{equation} \label{eq:derivative_chisquared}
	\frac{d}{dy} f_{\chi^2_{n}(\rho^2)}(y) = \frac{1}{2} f_{\chi^2_{n-2}(\rho^2)}(y)-\frac{1}{2} f_{\chi^2_{n}(\rho^2)}(y);
	\end{equation}
	and \eqref{eq:recurrence_relation} follows from the recurrence relationship
	\begin{equation}
	\sfI_{\nu-1}(z)-\sfI_{\nu+1}(z)=\frac{2\nu}{z}\sfI_\nu(z).
	\end{equation}
	Putting together \eqref{eq:afterQ} and \eqref{eq:recurrence_relation} we get
	\begin{align}
	&i_j(\rho_1;P_{\bfX})-i_j(\rho_2;P_{\bfX}) \\
	 &= -\frac{\rho_1^2-\rho_2^2}{2\sigma_j^2} \bbE\left[\bbE\left[\frac{\|\bfX\|}{\|\bfY\|}\sfh_{\frac{n}{2}}\left(\frac{\|\bfX\| \|\bfY\|}{\sigma_j^2}\right)-1| \frac{\left\|\bfY\right\|^2}{\sigma_j^2}=Q_j\right]\right].
	\end{align}
	We are now in the position to compute the derivative of the information density as:
	\begin{align}
	&i_j'(\rho;P_{\bfX}) \nonumber\\
	 &= \lim_{h\rightarrow 0} \frac{i_j(\rho+h;P_{\bfX})-i_j(\rho;P_{\bfX})}{h} \\
	&=-\frac{\rho}{\sigma_j^2}\: \bbE\left[\bbE\left[\frac{\|\bfX\|}{\|\bfY\|}\sfh_{\frac{n}{2}}\left(\frac{\|\bfX\| \|\bfY\|}{\sigma_j^2}\right)-1 | \frac{\left\|\bfY\right\|^2}{\sigma_j^2}=Q'\right]\right]
	\end{align}
	where $Q'\sim \chi^2_{n+2}(\frac{\rho^2}{\sigma_j^2})$ thanks to  Lemma \ref{Lemma:fQj}.
	
	The final result is obtained by letting
	\begin{align}
	\Xi'(\|\bfx \|;P_{\bfX}) = i_1'(\|\bfx \|;P_{\bfX})-i_2'(\|\bfx \|;P_{\bfX})
	\end{align}
	and by specializing the result to the input $P_{\bfX_{\sfR}}$.
\end{proof}
\begin{lemma}\label{Lemma:fQj}
	We have
	\begin{equation}
	\lim_{h\rightarrow 0} f_{Q_j}(y;\rho+h,\rho) = f_{\chi^2_{n+2}(\frac{\rho^2}{\sigma_j^2})}(y), \qquad y>0.
	\end{equation}
\end{lemma}
\begin{proof}
	Thanks to the definition \eqref{eq:def_fQ}, we have
	\begin{align}
	&\lim_{h\rightarrow 0} f_{Q_j}(y;\rho+h,\rho) \nonumber\\
	 &= \lim_{h\rightarrow 0} \frac{\sigma_j^2}{h(2\rho+h)}\left(F_{\chi^2_{n}(\frac{\rho^2}{\sigma_j^2})}(y)-F_{\chi^2_{n}(\frac{(\rho+h)^2}{\sigma_j^2})}(y)\right) \\
	&=\lim_{h\rightarrow 0} \frac{\sigma_j^2}{h(2\rho+h)} \int_{0}^{y} \left(f_{\chi^2_{n}(\frac{\rho^2}{\sigma_j^2})}(t)-f_{\chi^2_{n}(\frac{(\rho+h)^2}{\sigma_j^2})}(t)  \right) dt \\
	&= \frac{\sigma_j^2}{2\rho}\int_{0}^{y} \sum_{i=0}^{\infty} \lim_{h\rightarrow 0} \frac{1}{h} \nonumber\\
	&\quad\left(\frac{\rme^{-\frac{\rho^2}{2\sigma_j^2}}\left(\frac{\rho^2}{2\sigma_j^2}\right)^i}{i!}-\frac{\rme^{-\frac{(\rho+h)^2}{2\sigma_j^2}}\left(\frac{(\rho+h)^2}{2\sigma_j^2}\right)^i}{i!} \right) f_{\chi^2_{n+2i}}(t) dt \label{eq:Poisson_representation} \\
	&= \frac{\sigma_j^2}{2\rho}\int_{0}^{y} \sum_{i=0}^{\infty} \frac{d}{d\rho} \left(\frac{\rme^{-\frac{\rho^2}{2\sigma_j^2}}\left(\frac{\rho^2}{2\sigma_j^2}\right)^i}{i!}\right) f_{\chi^2_{n+2i}}(t) dt \\
	&= \frac{1}{2}\int_{0}^{y} \sum_{i=0}^{\infty}  \left(-\frac{\rme^{-\frac{\rho^2}{2\sigma_j^2}}\left(\frac{\rho^2}{2\sigma_j^2}\right)^i}{i!}\right.\nonumber\\
	&\qquad\qquad\left. +\frac{\rme^{-\frac{\rho^2}{2\sigma_j^2}}\left(\frac{\rho^2}{2\sigma_j^2}\right)^{i-1}}{(i-1)!}{1}(i\ge 1)\right) f_{\chi^2_{n+2i}}(t) dt	\\
	&= \frac{1}{2}\int_{0}^{y}  \left(-f_{\chi^2_{n}(\frac{\rho^2}{\sigma_j^2})}(t)+f_{\chi^2_{n+2}(\frac{\rho^2}{\sigma_j^2})}(t)\right)  dt \\
	&=\int_{0}^{y}  \frac{d}{dt}f_{\chi^2_{n+2}(\frac{\rho^2}{\sigma_j^2})}(t)  dt \label{eq:der_chi_squared} \\
	&=f_{\chi^2_{n+2}(\frac{\rho^2}{\sigma_j^2})}(y),
	\end{align}
	where ${1}(\cdot)$ is the indicator function; in \eqref{eq:Poisson_representation} we used the Poisson-weighted mixture representation of the noncentral chi-square pdf; and in \eqref{eq:der_chi_squared} we used \eqref{eq:derivative_chisquared}.
\end{proof}

\bibliographystyle{IEEEtran}
\bibliography{refs.bib}

\begin{thebibliography}{10}
\providecommand{\url}[1]{#1}
\csname url@samestyle\endcsname
\providecommand{\newblock}{\relax}
\providecommand{\bibinfo}[2]{#2}
\providecommand{\BIBentrySTDinterwordspacing}{\spaceskip=0pt\relax}
\providecommand{\BIBentryALTinterwordstretchfactor}{4}
\providecommand{\BIBentryALTinterwordspacing}{\spaceskip=\fontdimen2\font plus
\BIBentryALTinterwordstretchfactor\fontdimen3\font minus
  \fontdimen4\font\relax}
\providecommand{\BIBforeignlanguage}[2]{{%
\expandafter\ifx\csname l@#1\endcsname\relax
\typeout{** WARNING: IEEEtran.bst: No hyphenation pattern has been}%
\typeout{** loaded for the language `#1'. Using the pattern for}%
\typeout{** the default language instead.}%
\else
\language=\csname l@#1\endcsname
\fi
#2}}
\providecommand{\BIBdecl}{\relax}
\BIBdecl

\bibitem{wyner1975wire}
A.~D. Wyner, ``The wire-tap channel,'' \emph{Bell Syst. Tech. J.}, vol.~54,
  no.~8, pp. 1355--1387, 1975.

\bibitem{bloch2011physical}
M.~Bloch and J.~Barros, \emph{Physical-Layer Security:From Information Theory
  to Security Engineering}.\hskip 1em plus 0.5em minus 0.4em\relax Cambridge
  University Press, 2011.

\bibitem{Oggier2015Wiretap}
F.~Oggier and B.~Hassibi, ``A perspective on the {MIMO} wiretap channel,''
  \emph{Proc. of {IEEE}}, vol. 103, no.~10, pp. 1874--1882, 2015.

\bibitem{Liang2009Security}
Y.~Liang, H.~V. Poor, and S.~{Shamai (Shitz)}, ``Information theoretic
  security,'' \emph{Foundations and Trends in Communications and Information
  Theory}, vol.~5, no. 4--5, pp. 355--580, 2009.

\bibitem{poor2017wireless}
H.~V. Poor and R.~F. Schaefer, ``Wireless physical layer security,''
  \emph{Proc. the Natl. Acad. Sci. U.S.A.}, vol. 114, no.~1, pp. 19--26, 2017.

\bibitem{GaussianWireTap}
S.~Leung-Yan-Cheong and M.~Hellman, ``The {G}aussian wire-tap channel,''
  \emph{IEEE Trans. Inf. Theory}, vol.~24, no.~4, pp. 451--456, 1978.

\bibitem{khisti2010secure}
A.~Khisti and G.~W. Wornell, ``Secure transmission with multiple
  antennas--{P}art {II}: {T}he {MIMOME} wiretap channel,'' \emph{IEEE
  Transactions on Information Theory}, vol.~56, no.~11, pp. 5515--5532, 2010.

\bibitem{oggier2011secrecy}
F.~Oggier and B.~Hassibi, ``The secrecy capacity of the {MIMO} wiretap
  channel,'' \emph{IEEE Transactions on Information Theory}, vol.~57, no.~8,
  pp. 4961--4972, 2011.

\bibitem{I-MMSE}
D.~Guo, S.~Shamai, and S.~Verd{\'u}, ``Mutual information and minimum
  mean-square error in {G}aussian channels,'' \emph{IEEE Trans. Inf. Theory},
  vol.~51, no.~4, pp. 1261--1282, 2005.

\bibitem{bustin2009mmse}
R.~Bustin, R.~Liu, H.~V. Poor, and S.~Shamai, ``An {MMSE} approach to the
  secrecy capacity of the {MIMO} {G}aussian wiretap channel,'' \emph{EURASIP
  Journal on Wireless Communications and Networking}, vol. 2009, pp. 1--8,
  2009.

\bibitem{ozel2015gaussian}
O.~Ozel, E.~Ekrem, and S.~Ulukus, ``Gaussian wiretap channel with amplitude and
  variance constraints,'' \emph{IEEE Trans. Inf. Theory}, vol.~61, no.~10, pp.
  5553--5563, 2015.

\bibitem{barletta2021scalar}
L.~Barletta and A.~Dytso, ``Scalar {G}aussian wiretap channel: Bounds on the
  support size of the secrecy-capacity-achieving distribution,'' in \emph{2021
  IEEE Information Theory Workshop (ITW)}, 2021, pp. 1--6.

\bibitem{DytsoITWwiretap2018}
A.~Dytso, M.~Egan, S.~M. Perlaza, H.~V. Poor, and S.~S. Shitz, ``Optimal inputs
  for some classes of degraded wiretap channels,'' in \emph{2018 IEEE
  Information Theory Workshop (ITW)}, 2018, pp. 1--5.

\bibitem{dytsoMI_est_2019}
A.~Dytso, M.~Al, H.~V. Poor, and S.~Shamai~Shitz, ``On the capacity of the peak
  power constrained vector gaussian channel: An estimation theoretic
  perspective,'' \emph{IEEE Transactions on Information Theory}, vol.~65,
  no.~6, pp. 3907--3921, 2019.

\bibitem{favano2021capacity}
A.~Favano, M.~Ferrari, M.~Magarini, and L.~Barletta, ``The capacity of the
  amplitude-constrained vector {G}aussian channel,'' in \emph{2021 IEEE
  International Symposium on Information Theory (ISIT)}, 2021, pp. 426--431.

\bibitem{berry1990minimax}
J.~C. Berry, ``Minimax estimation of a bounded normal mean vector,''
  \emph{Journal of Multivariate Analysis}, vol.~35, no.~1, pp. 130--139, 1990.

\bibitem{barletta2021numerical}
L.~Barletta and A.~Dytso, ``Scalar {G}aussian wiretap channel with peak
  amplitude constraint: {N}umerical computation of the optimal input
  distribution,'' \emph{arXiv preprint arXiv:2111.11442}, 2021.

\bibitem{karlin1957polya}
S.~Karlin, ``P{\'o}lya type distributions, ii,'' \emph{The Ann. Math. Stat.},
  vol.~28, no.~2, pp. 281--308, 1957.

\bibitem{verdu2010mismatched}
S.~Verd{\'u}, ``Mismatched estimation and relative entropy,'' \emph{IEEE
  Transactions on Information Theory}, vol.~56, no.~8, pp. 3712--3720, 2010.

\bibitem{segura2011bounds}
J.~Segura, ``Bounds for ratios of modified {B}essel functions and associated
  {T}ur{\'a}n-type inequalities,'' \emph{Journal of Mathematical Analysis and
  Applications}, vol. 374, no.~2, pp. 516--528, 2011.

\bibitem{baricz2015bounds}
{\'A}.~Baricz, ``Bounds for {T}ur{\'a}nians of modified {B}essel functions,''
  \emph{Expositiones Mathematicae}, vol.~33, no.~2, pp. 223--251, 2015.

\bibitem{lapidoth2003capacity}
A.~Lapidoth and S.~M. Moser, ``Capacity bounds via duality with applications to
  multiple-antenna systems on flat-fading channels,'' \emph{IEEE Transactions
  on Information Theory}, vol.~49, no.~10, pp. 2426--2467, 2003.

\end{thebibliography}
\end{document}